\documentclass{article}
\usepackage{spconf,amsmath,graphicx}
\usepackage{amssymb,amsthm}
\usepackage{hyperref}
 \DeclareGraphicsExtensions{.pdf}
 \graphicspath{{./}{./pics/}}

\newcommand{\E}{\mathbb{E}}                  
\newcommand{\bo}[1]{\mathbf{#1}}              
\newcommand{\bom}[1]{\boldsymbol{#1}}    

\newcommand{\be}{\beta}
  
\newcommand{\ka}{\kappa}

\newcommand{\pdim}{p}
\newcommand{\ndim}{n}
\renewcommand{\Pi}{{\mathcal P}}

\newcommand{\beq}{\begin{equation}}
\newcommand{\eeq}{\end{equation}}
\newcommand{\bmat}{\begin{pmatrix}}
\newcommand{\emat}{\end{pmatrix}}

\newcommand{\I}{\bo I}

\newcommand{\A}{{\bo A}}

\renewcommand{\S}{{\bo S}} 

\newcommand{\x}{\bo x}

\newcommand{\Fr}{\mathrm{F}}

\newcommand{\M}{\bom \Sigma}  
\newcommand{\Mn}{\M_0}  	    
\newcommand{\Mor}{\bo C}	    

\newcommand{\MSE}{\mathrm{MSE}}

\newcommand{\sca}{\sigma}

\DeclareMathOperator{\tr}{tr}
\DeclareMathOperator{\Tr}{tr}

\DeclareMathOperator{\cov}{cov}

\newcounter{ctheorem}
\newtheorem{theorem}[ctheorem]{Theorem}

\newcounter{clemma}
\newtheorem{lemma}[clemma]{Lemma}
\newcounter{ccorollary}
\newtheorem{corollary}[ccorollary]{Corollary}
\title{M-estimators of scatter with eigenvalue shrinkage}
\name{Esa Ollila$^{\star }$ \qquad Daniel P. Palomar$^{\dagger}$ \qquad  Fr\'ed\'eric Pascal$^{\ddagger}$}
\address{$^{\star}$Department of Signal Processing and Acoustics, Aalto University, Finland \\    
$^{\dagger}$ The Hong Kong University of Science and Technology, Hong Kong \\
$^{\ddagger}$ L2S / CentraleSupel\'ec, University Paris-Saclay, France}
\begin{document}
%
\maketitle
\begin{abstract}
A popular regularized (shrinkage) covariance estimator is the shrinkage sample covariance matrix (SCM)  which shares the same set of eigenvectors as the SCM but shrinks its eigenvalues toward its grand mean. In this paper, a more general approach is considered in which  the SCM is replaced  by an M-estimator of scatter matrix and a fully automatic data adaptive method to compute the optimal shrinkage parameter with minimum mean squared error is proposed. Our approach permits the use of any weight function such as Gaussian, Huber's, or $t$ weight functions, all of which are commonly used in M-estimation framework. Our simulation examples illustrate that shrinkage M-estimators based on the proposed optimal tuning combined with robust weight function do not loose in performance to shrinkage SCM estimator when the data is Gaussian, but provide significantly improved performance when the data is sampled from a heavy-tailed distribution. 
\end{abstract}
\section{Introduction} \label{sec:intro}

Consider a sample of $p$-dimensional vectors $\{ \x_i \}_{i=1}^n$ sampled from a distribution of a random vector $\x$ with $\E[\x]= \bo 0$.   One of the first tasks in the analysis of high-dimensional data is to estimate the covariance matrix.   The most commonly used estimator is the sample covariance matrix (SCM), $\S=\frac 1 n \sum_{i=1}^\ndim \x_i \x_i^\top$, 
but its main drawbacks are  its loss of efficiency when sampling from distributions which have longer
tails than the multivariate normal (MVN) distribution and its sensitivity to outliers.  Although being unbiased estimator of 
the covariance matrix $\cov(\x) = \E[\x \x^\top]$ for any sample length $n \geq 1$, it is well-known that the eigenvalues are poorly estimated when $n$ is not orders of magnitude larger than $p$. In such cases, one commonly uses a  regularized SCM (RSCM) with a linear shrinkage towards a scaled identity matrix, 
\beq \label{eq:LW}
\S_\be = \be \S + (1-\be) \frac{\Tr(\S)}{p} \I,
\eeq 
where $\be \in (0,1]$ is the regularization parameter.  The RSCM $\S_\be$ shares the same set of eigenvectors as the SCM $\S$, but its eigenvalues are shrinked towards the grand mean of the eigenvalues. That is, if $d_1,\ldots, d_p$ denote the eigenvalues of $\bo S$, then $ \be d_j + (1-\be)\bar d$ are the eigenvalues of  $\bo S_\be$, where $\bar d=p^{-1}\sum_j d_j$. Optimal computation of  $\be$ such that $\S_\be$ has minimum mean squared error (MMSE) has been developed for example in  \cite{ledoit2004well,ollila2019optimal}. 

The estimator in \eqref{eq:LW} remains sensitive to outliers and non-Gaussianity.  M-estimators of scatter \cite{maronna1976robust} are popular robust alternatives to SCM. We consider the situation where  $n>p$ and hence a conventional M-estimator of scatter  $\hat \M$ exists and can be used  in place of the SCM $\bo S$  in  \eqref{eq:LW}.  We then propose a fully automatic data adaptive method to  compute the optimal shrinkage parameter $\be$.    First, we derive an approximation for parameter $\be$ that attains the minimum MMSE and then propose a data adaptive method for its computation.  The benefit of our approach is that it can be easily applied to any M-estimator using any weight function $u(t)$.   Our simulation examples illustrate that a shrinkage M-estimator using the proposed tuning and a robust loss function does not loose in performance to  optimal shrinkage SCM estimator when the data is Gaussian, but is able to provide significantly improved performance in the case of heavy-tailed data.  

{\it Relations to prior work}: Earlier work, \cite{ollila2014regularized,pascal2014generalized, sun2014regularized,chen2010shrinkage,couillet2014large,zhang2016automatic},  proposed regularized M-estimators of scatter matrix either  by  adding a penalty function to M-estimation objective function or a diagonal loading term to the respective first-order solution (M-estimating equation).  We consider  a simpler approach that uses conventional M-estimator and shrinks its  eigenvalues to grand mean of the eigenvalues.  Our approach permits  computation of the optimal shrinkage parameter for any M-estimation weight function. 

The paper is structured as follows. Section~\ref{sec:shrinkM} introduces the proposed  shrinkage M-estimator framework. 
Section~\ref{sec:beta_opt}  discusses automatic computation of the optimal shrinkage parameter under the assumption of sampling from unspecified elliptical distribution.  Section~\ref{sec:simul} contains simulation studies.

\section{Shrinkage M-estimators of scatter} \label{sec:shrinkM}

In this paper, we assume that $n>p$ and consider an M-estimator of scatter matrix \cite{maronna1976robust} that 
solves an estimating equation
\beq \label{eq:Mest}
\hat \M = \frac{1}{n} \sum_{i=1}^n u(\x_i^\top \hat \M^{-1} \x_i) \x_i \x_i^\top ,
\eeq 
where $u: [0, \infty) \to [0,\infty)$ is a non-increasing weight function.  An M-estimator is a sort of adaptively weighted sample covariance matrix with weights determined by function $u(\cdot)$.   To guarantee existence of the solution,  it is required that the data verifies the condition stated in   \cite{kent1991redescending}. An M-estimator of scatter which shrinks the eigenvalues towards the grand mean of the eigenvalues is then defined as:  
\beq \label{eq:shrinkMest}
\hat \M_\be = \be \hat \M + (1-\be) \frac{\Tr(\hat \M)}{p} \I. 
\eeq 
 For example, the RSCM $\bo S_\be$ is obtained when one uses the Gaussian weight function $u(t)=1$ $\forall t$ since then $\hat \M = \bo S$. Other popular choices are Huber's weight function
\beq \label{eq:huber_u}
u_{\mbox{\tiny H}} (t;c) =  \max(-c^2,\min(t,c^2))/b,
\eeq 
where   $c > 0$ is a tuning constant, chosen by the user, and $b$ is  a scaling factor used to obtain Fisher consistency at the multivariate normal (MVN) distribution $\mathcal N_\pdim(\bo 0 ,\M)$:
\[
b = F_{\chi^2_{\pdim+2}}(c^2) + c^2(1-F_{\chi^2_{\pdim}}(c^2))/\pdim.
 \]
We choose $c^2$ as $q$th upper quantile of $\chi^2_\pdim$: $c^2 = F^{-1}_{\chi^2_p}(q)$. 
Another popular choice is  $t$-MLE weight function 
\beq \label{eq:t_u}
u_{\mbox{\tiny T}}(t ; \nu) = \frac{\pdim  + \nu }{\nu +t}
\eeq 
in which case the corresponding M-estimator $\hat \M$  is also the maximum likelihood estimate (MLE) of the scatter matrix parameter of a  $p$-variate $t$-distribution with $\nu>0$ degrees of freedom. 

An M-estimator is a consistent estimator of the M-functional of scatter matrix,  defined as 
\vspace*{-2pt} \beq \label{eq:Mfun}
\Mn = \E\big[ u(\x^\top \Mn^{-1} \x)\x\x^\top \big].
\vspace*{-2pt} \eeq  
If the population M-functional $\Mn$ is known,  then by defining a  \emph{1-step estimator}
\vspace*{-5pt} \beq \label{eq:Mor} 
\Mor =    \frac{1}{n} \sum_{i=1}^n u(\x_i^\top \Mn^{-1} \x_i) \x_i \x_i^\top   
\vspace*{-5pt} \eeq
we can compute 
\vspace*{-2pt} \beq\label{eq:Mor2}
 \bo C_\be =   \be  \Mor + (1-\be)[\tr(\bo C)/p] \I 
\vspace*{-2pt} \eeq 
as  a proxy for $\hat \M_\be$.  Naturally, such an estimator is fictional, as the initial value $\Mn$ is unknown. 
 The 1-step estimator $\Mor$ is an unbiased estimator of  $\Mn$,  i.e., $\E[\Mor] = \Mn$. 

Ideally we would like to find the value of  $\beta  \in [0,1]$ for which  the corresponding estimator $\hat \M_\be$ attains the minimum MSE, that is, 
\beq \label{eq:optimal_beta_o}
\beta_o = \arg \min_\be  \Big\{ \mathrm{MSE}(\hat \M_\be) =  \mathbb{E} \Big[ \big\| \hat \M_\be - \M_0 \|^2_{\mathrm{F}} \Big] \Big\},
\eeq 
where $\| \cdot \|_{\Fr}$ denotes the Frobenius matrix norm ($\| \A \|_{\Fr}^2 = \tr(\bo A^\top \bo A)$).  
Since solving \eqref{eq:optimal_beta_o} is not doable due to the implicit form of M-estimators, we look for an approximation:
\beq \label{eq:RFegSCM_oracle} 
\be_{o}^{\mathrm{app}} =   \underset{\be}{\arg \min}  \ \Big \{ \mathrm{MSE}(\bo C_\be) =  \E \Big[ \big\| \bo C_\be - \Mn \big\|^2_{\rm F}  \Big] \Big\} .   
\eeq 
Such approach was also used  in \cite{chen2011robust} in deriving an optimal parameter for shrinkage Tyler's M-estimator of scatter. 

Before stating the expression for $ \be_{o}^{\mathrm{app}}$ we  introduce a \emph{sphericity}  measure of scatter: 
\beq \label{eq:gamma} 
 \gamma =  \frac{p \tr(\Mn^2)}{\tr(\Mn)^2}   .
\eeq 
Sphericity $\gamma$  measures how close $\Mn$ is to a scaled identity matrix: $\gamma \in [1,p]$ 
where $\gamma=1$ if and only if $\Mn \propto \I$ and $\gamma = p$ if $\Mn$ has rank equal to 1.

\begin{theorem}  \label{th:beta0}    
Suppose $\x_1, \ldots, \x_n$  is an i.i.d.\ random sample from any $p$-variate distribution (not necessarily elliptical distribution),    and 
$u$ is a weight function for which the expectation $\E[ \tr(\bo C^2)]$ exists. 
The oracle  parameter $\be_o^{\mathrm{app}}$ in \eqref{eq:RFegSCM_oracle} 
 is
\begin{align}
\be_o^{\mathrm{app}}  &= \frac{\| \Mn - \eta_o \I \|_{\Fr}^2}{\E \big [ \big\| \bo C -   (\tr(\bo C)/p) \I \big\|_{\Fr}^2 \big]} \label{eq:beta0id} \\ 
	    &= \frac{p (\gamma-1) \eta^2_o}{\E[\tr(\bo C^2)] - p^{-1} \E[ \tr(\bo C)^2]} \label{eq:beta0id2} 
\end{align}
where $\eta_o  = \tr(\Mn)/p$ and $\gamma$ is defined in \eqref{eq:gamma}. Note that $\be_o^\mathrm{app} \in \left[0,1\right)$ and the value of the MSE
at the optimum  is
\begin{equation} \label{eq:MSEopt}
 \MSE(\bo C_{\be_o^{\mathrm{app}}})  =  \frac{ \E[\tr(\bo C)^2 ] - \tr(\Mn)^2}{p}+ ( 1-  \beta_o^{\mathrm{app}})  \big\| \Mn - \eta_o  \big \|^2_{\Fr}.
\end{equation}
\end{theorem}
\begin{proof}
Write $L(\be) = \MSE(\bo C_\be) = \E[\| \bo C_\be -\Mn\|_{\Fr}^2]$. Then note that 
\begin{align*}
L(\be)
&= \E \big[  \big \| \be \bo C + (1- \be) p^{-1} \tr(\bo C) \I - \Mn \big\|^2_{\Fr} \big] \\ 
&= \E \Big[  \big \| \be(\bo C - \Mn)  + (1- \be) \big ( p^{-1} \tr(\bo C)  \I - \Mn \big) \big\|^2_{\Fr} \Big] \\   
&= \be^2 a_1 +  (1- \be)^2  a_2 + 2\be(1-\be) a_3, 
\end{align*}
where 
 $a_1 =\E \big[\big\| \bo C - \Mn \big\|_{\Fr}^2 \big] =  \E \big[ \tr(\bo C^2) \big] - \tr(\Mn^2)$, and
 \begin{align*}
a_2 &= \E \big [ \big\|  p^{-1} \tr(\bo C) \I  - \Mn  \big\|^2_{\Fr} \big] \\
&= a_3+ \tr(\Mn^2) -  p \eta_o^2  = a_3 + p ( \gamma -1)\eta_o^2  \\ 
a_3 &= p^{-1} \E \big[  \tr(\bo C) \tr(\bo C - \Mn) \big]  
= p^{-1}\E\big[	\tr (\bo C)^2  \big]  - \eta_o^2 p .
\end{align*} 
Note that $L(\be)$ is a convex quadratic function in $\be$ with a unique minimum given by 
\[
\be_o^\mathrm{app} =  \frac{a_2-a_3}{ (a_1-a_3) + (a_2-a_3)}.  
\]
Substituting the expressions for constants $a_1, a_2$ and $a_3$ into $\be_o^\mathrm{app} $ yields the stated result. 
\end{proof}

Next we derive a more explicit form of  $\be_o^{\mathrm{app}}$ by assuming that the data is generated from unspecified elliptically symmetric distribution. 

\section{Shrinkage parameter computation}  \label{sec:beta_opt}

Maronna \cite{maronna1976robust} developed M-estimators of scatter matrices originally  within the framework of elliptically symmetric distributions \cite{fang1990symmetric,ollila2012complex}.   The probability density function (p.d.f.) of centered (zero mean) elliptically distributed
random vector $\x \sim \mathcal E_\pdim(\bo 0,\M,g)$ is
\[
	f(\x)
	= C_{\pdim,g} |\M|^{-1/2} g\big(\x^\top \M^{-1} \x\big),
\]
where  $\M$ is the positive definite symmetric  matrix parameter, called the scatter matrix, $g: \left[0,\infty\right) \to \left[0,\infty\right)$ is the
\emph{density generator}, which is a fixed function that is independent of
$\x$ and $\M$, and $C_{\pdim,g}$ is a normalizing constant ensuring
that $f(\x)$ integrates to 1. 
The density generator $g$ determines
the elliptical distribution. For example, the MVN distribution is obtained when
$g(t)=\exp(-t/2)$ and the $t$-distribution with $\nu$ d.o.f., denoted $\x \sim t_{\nu}(\bo 0, \M,g)$, is obtained when $g(t)= (1+  t/\nu)^{-(p+\nu)/2}$. Then the weight function  for the MLE of scatter corresponds to  the case that 
$
u(t) \propto-   g'(t)/g(t).
$
 This yields  \eqref{eq:t_u} as the weight function for 
the MLE of scatter when $\x \sim t_{\nu}(\bo 0, \M,g)$.  If the second moments of $\x$ exists, then $g$ can be defined so that $\M$ represents the covariance matrix of $\x$, i.e., $\M=\cov(\x)$; see  \cite{ollila2012complex} for details.

When $\x \sim \mathcal E_\pdim(\bo 0,\M,g)$, then the M-functional $\Mn$ in \eqref{eq:Mfun} is related to underlying scatter matrix parameter $\M$  via the relationship
\[
\Mn=\sca \M, 
\]
where $\sca>0$ is a solution to an equation
\begin{equation}\label{eq:sca} 
\E \bigg[ \psi \bigg( \frac{ \x^\top \M^{-1} \x }{\sca} \bigg)  \bigg] = p,  
\end{equation}
 where $\psi(t)=u(t)t$.  
Often $\sca$  needs to be  solved numerically from \eqref{eq:sca} but in some cases an analytic expression 
can be derived. If  $\x \sim \mathcal E_p(\bo 0,\M,g)$ and the used weight function matches with the data distribution, so 
$u(t) \propto-   g'(t)/g(t)$, then $\sigma = 1$.

Next we derive expressions for  $ \E[\tr(\bo C)^2 ]$ and $ \E[\tr(\bo C^2)]$  appearing in the denominator of $\beta_o^{\mathrm{app}}$ in \eqref{eq:beta0id2}. 
They depend on a constant $\psi_1$ (which depend on weight function $u$ via $\psi(t) = u(t) t$) as follows:
\beq \label{eq:psi1}
\psi_1 = \frac{1}{p(p+2) } \E \Big[   \psi \!\bigg( \frac{ \x^\top \M^{-1} \x }{\sca}\Big)^2   \bigg]  ,
\eeq
 where the expectation is w.r.t.   $\x  \sim \mathcal E_p(\bo 0, \M,g)$. 

\begin{lemma} \label{lem:EtrC} Suppose $\x_1, \ldots, \x_n$  is an i.i.d.\ random sample from $\mathcal E_p(\bo 0, \M,g)$. Then 
\begin{align*}
\E \big[\tr \! \big(\Mor^2\big) \big]  &= 
\left( 1 + \frac{2 \psi_1-1}{n} \right) \tr(\Mn^2) +  \frac{\psi_1}{n}\tr(\Mn)^2   ,  \\
\E[\tr(\bo C)^2] 
	&=  \frac{2 \psi_1}{n}   \tr(\Mn^2) +   \Big(1+ \frac{\psi_1-1}{n}\Big) \tr(\Mn)^2, 
\end{align*}
given that expectation \eqref{eq:psi1} exists. 
\end{lemma}

\begin{proof} Omitted due to lack of space. 
\end{proof} 

\begin{theorem}\label{th:beta0ell}
Let $\x_{1},\ldots,\x_{n}$ denote an i.i.d. random sample from an
elliptical distribution $\mathcal E_p(\bo 0, \M,g)$.  Then the oracle parameter $\be_o^\mathrm{app}$ that
minimizes the MSE in Theorem~1 is
\begin{align} \label{eq:beta0ell}
 \be_o^{\mathrm{app}}    &=    \dfrac{ \gamma-1}{  (\gamma  -1)(1-1/n)  + \psi_1(1-1/p)(2 \gamma+p)/n}    
\end{align}
 where   $\gamma$ is the sphericity measure defined in~\eqref{eq:gamma}.
 \end{theorem}

\begin{proof} Follows from Theorem~\ref{th:beta0} after substituting the values for $\E \big[\tr \! \big(\Mor^2\big) \big]$ and  $\E[\tr(\bo C)^2]$  
derived in Lemma~\ref{lem:EtrC} in the denominator of $\be_o^{\mathrm{app}}$ in \eqref{eq:beta0id2}. 
\end{proof}

If one uses Gaussian loss function $u(t) \equiv 1$, then one needs to assume that the 4th-order moments exists and one may assume w.l.o.g. that the scatter matrix parameter equals the covariance matrix \cite{ollila2012complex}, i.e.,  $\M= \cov(\x)$ , so  $\Mn=  \M$ and $\sigma=1$. Furthermore, it holds that  $\hat \M= \S$ and  $\bo C_\be=\S_\be$ and hence $\be_o  = \be_o^\mathrm{app}$. Finally,  we may relate $\psi_1$ with an elliptical kurtosis~\cite{muirhead:1982} parameter $\kappa$, defined as
\vspace*{-7pt} \beq \label{eq:kappa} 
\ka=  \frac{\E [ \| \M^{-1/2}\x \|^4]}{p(p+2)}  - 1 .
\vspace*{-5pt} \eeq  
Elliptical kurtosis  vanishes, i.e., $\ka=0$, when $\x  \sim \mathcal N_p(\bo 0, \M)$.  

\begin{corollary} Let $\x_{1},\ldots,\x_{n}$ denote an i.i.d. random sample from an
elliptical distribution $\mathcal E_p(\bo 0, \M,g)$ with finite 4th order moments and covariance matrix $\M=\cov(\x)$. 
Then  the optimal tuning parameter of the shrinkage SCM estimator $\S_\be$ in \eqref{eq:LW} is 
\vspace*{-9pt}\begin{align} 
\be_o   =  \arg \min_\be  \, \mathbb{E} \big[ \big\| \S_\be - \M \|^2_{\mathrm{F}} \big] =\dfrac{\gamma-1}{  \gamma-1 +  a  }  ,\vspace*{-8pt} \label{eq:beta0_ell_scm}
\end{align} 
\vspace*{-5pt} where
\vspace*{-5pt}\[
a = \frac{\ka(2 \gamma(1-1/p) + p-1)}{n}   + \frac{ \gamma (1-2/p)+p}{n}.
\vspace*{-2pt}\] 
\end{corollary}

\begin{proof}  The result follows from Theorem~\ref{th:beta0ell} since $\bo C_\be = \bo S_\be$ and the M-functional for Gaussian loss is  $\Mn = \cov(\x) = \M$.  Since for Gaussian loss, $\psi(t)= t$, we notice from  \eqref{eq:psi1}  that $\psi_1 = 1 + \kappa$.  Plugging $\psi_1=1+\ka$ into \eqref{eq:beta0ell}  yields  the stated expression. 
\end{proof}

\section{Simulation studies}  \label{sec:simul}

We  compute different shrinkage M-estimators  $\hat \M_\be$ detailed below.   We use acronym {\bf Huber} to refer to the shrinkage M-estimator $\hat \M_\be$ that uses Huber's weight $u(\cdot) = u_{\mbox{\tiny H}}(\cdot; c) $ with  threshold  $c^2$ corresponding to $q=0.7$ quantile.  Shrinkage parameter is computed as  $\beta= \beta_o^{\mathrm{app}}(\hat \gamma, \hat \psi_{1})$.  As an estimator $\hat \gamma$ of $\gamma$ we use the same estimate as in \cite{zhang2016automatic,ollila2019optimal} and  $\hat \psi_{1}$ is an estimate of $\psi_1$,  computed as 
\vspace*{-8pt}\beq \label{eq:phi1-est}
\hat \psi_{1} =  \frac 1 n \sum_{i=1}^n  \frac{ [t_i  u(t_i)]^2}{p (p+2)}  ,
\vspace*{-8pt}
\eeq 
where $t_i = \x_i^\top \hat \M^{-1} \x_i$ and  $\hat \M$ is the corresponding Huber's M-estimator.   Huber's weight function is standardized to be Fisher consistent for Gaussian samples, meaning that \eqref{eq:sca}  holds with $\sigma=1$ when $\x \sim \mathcal N_p(\bo 0, \M)$. 
Since \eqref{eq:phi1-est} ignores estimation of $\sigma$, some loss in accuracy of this estimate  of $\psi_1$ is  expected for non-Gaussian data.

Acronym {\bf t-MLE} refers to the shrinkage M-estimator  of scatter using  weight function $u(\cdot) = u_{\mbox{\tiny T}}(\cdot ; \nu)$, where d.o.f. parameter $\nu$ is estimated from the data. This means that  $\sigma=1$ can be assumed  since the scaling factor $\sigma$ equals one for an MLE of the scatter matrix parameter.  The shrinkage parameter is computed as  $\beta= \beta_o^{\mathrm{app}}(\hat \gamma, \hat \psi_{1})$, where  $\hat \gamma$ is as earlier  and  $\hat \psi_1$ is computed as in \eqref{eq:phi1-est} but using $u(\cdot) = u_{\mbox{\tiny T}}(\cdot ; \nu)$ and  $\hat \M$ being the corresponding M-estimator. 

Acronym {\bf Gauss} refers to the shrinkage M-estimator  of scatter using  Gaussian weight function $u(t)=1$, i.e.,  $\hat \M_\be=\S_{\be}$.  The shrinkage parameter is computed as $\be=\beta_o(\hat \kappa,\hat \gamma)$ with $\beta_o$ given by   \eqref{eq:beta0_ell_scm} and $\hat \kappa$ is an estimate  of elliptical kurtosis $\kappa$ proposed in  \cite{ollila2019optimal}.  Finally, acronym {\bf LW} refers to estimator proposed by Ledoit and Wold \cite{ledoit2004well}.  LW estimator also uses RSCM $\S_{\be}$, where parameter $\be$ is computed in a different manner than for Gauss estimator.

\begin{figure}[!t]
\centerline{\includegraphics[width=0.37\textwidth]{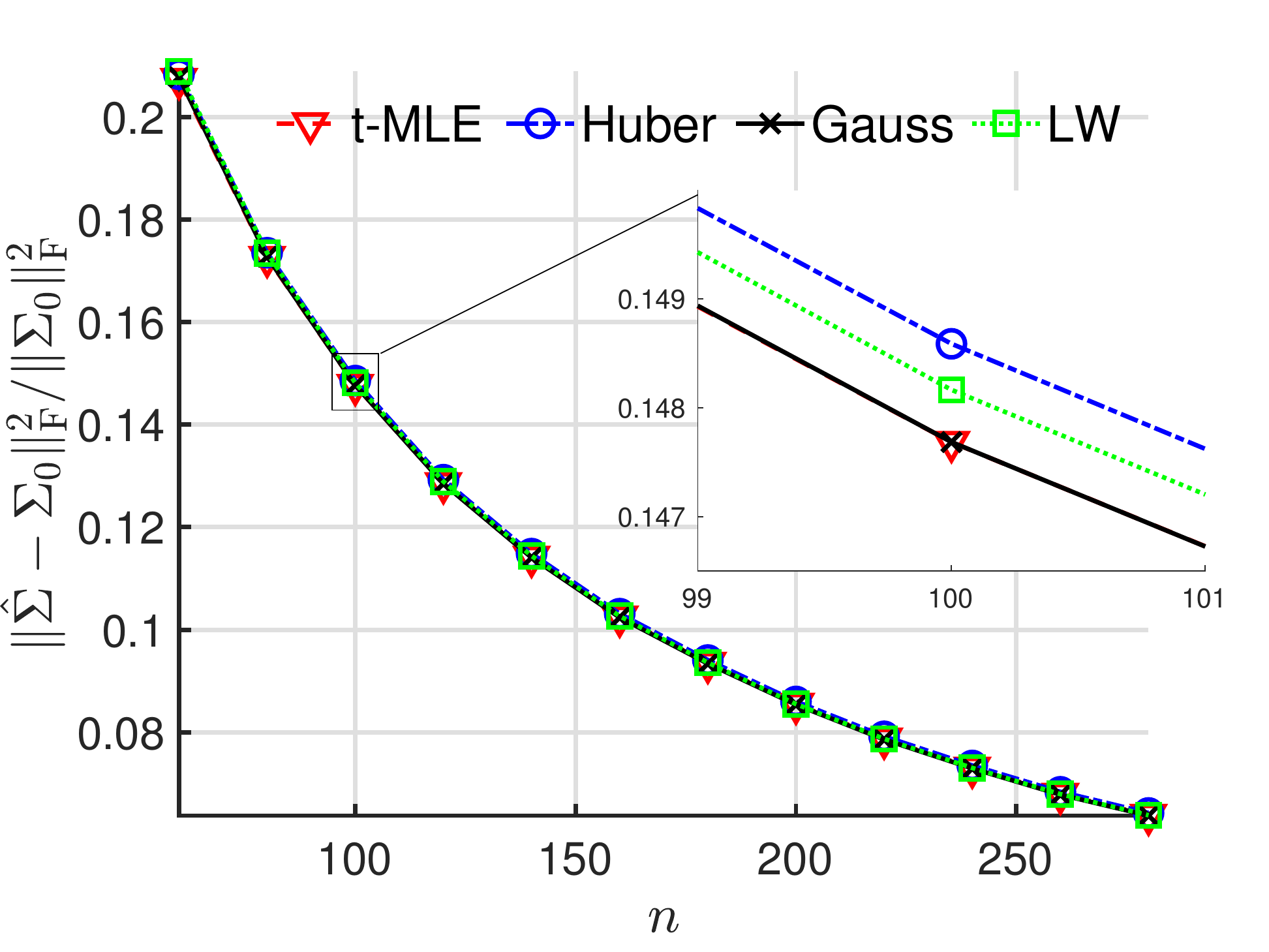}}
\vspace{-0.4cm}
\caption{\small NMSE of  the estimators a function of $n$ when samples are drawn from MVN distribution with an AR(1) covariance structure withs $\varrho=0.6$ and $p=40$.} \label{fig:AR1}
\centerline{\includegraphics[width=0.265\textwidth]{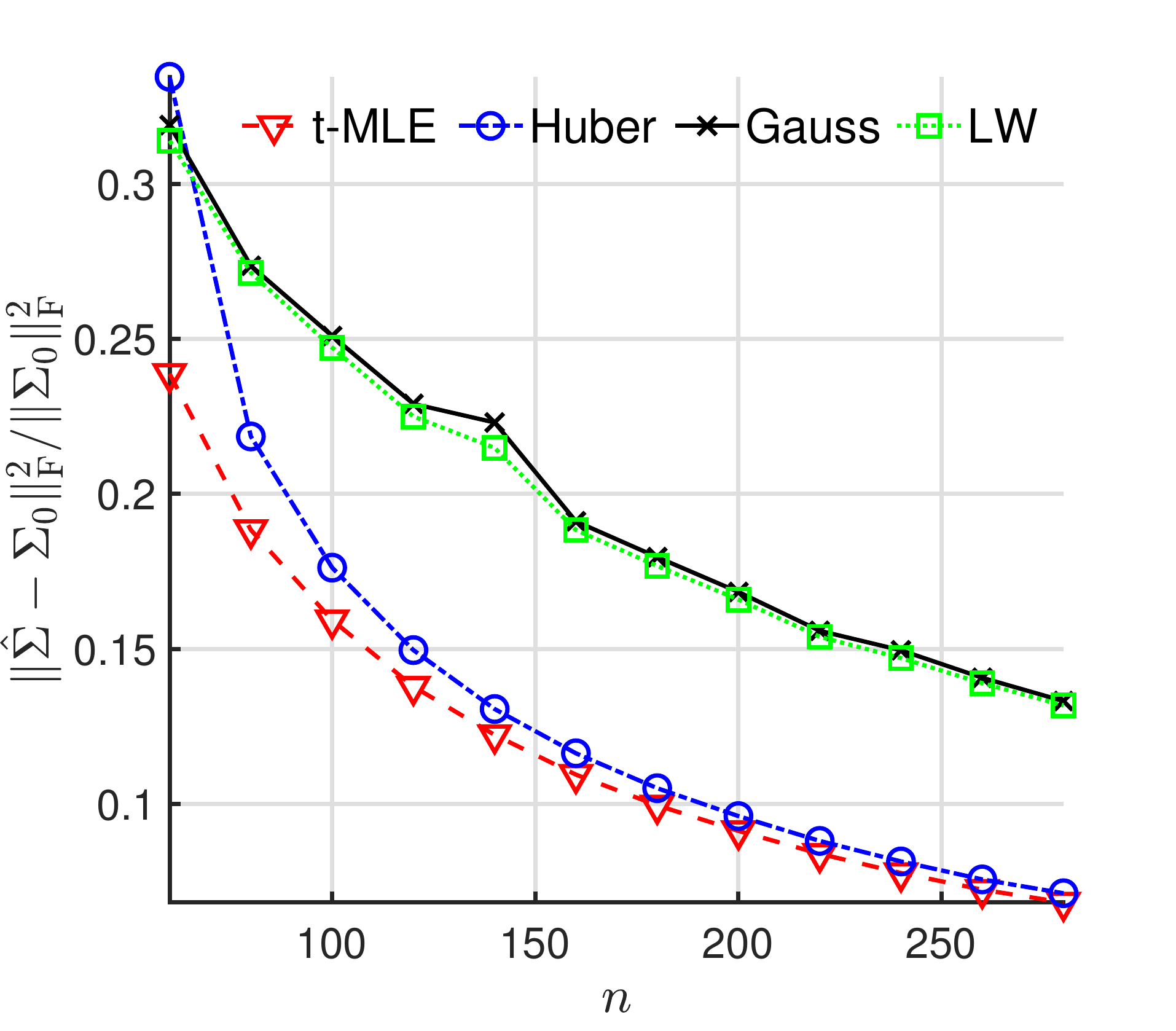}  \hspace{-0.7cm}
\includegraphics[width=0.265\textwidth]{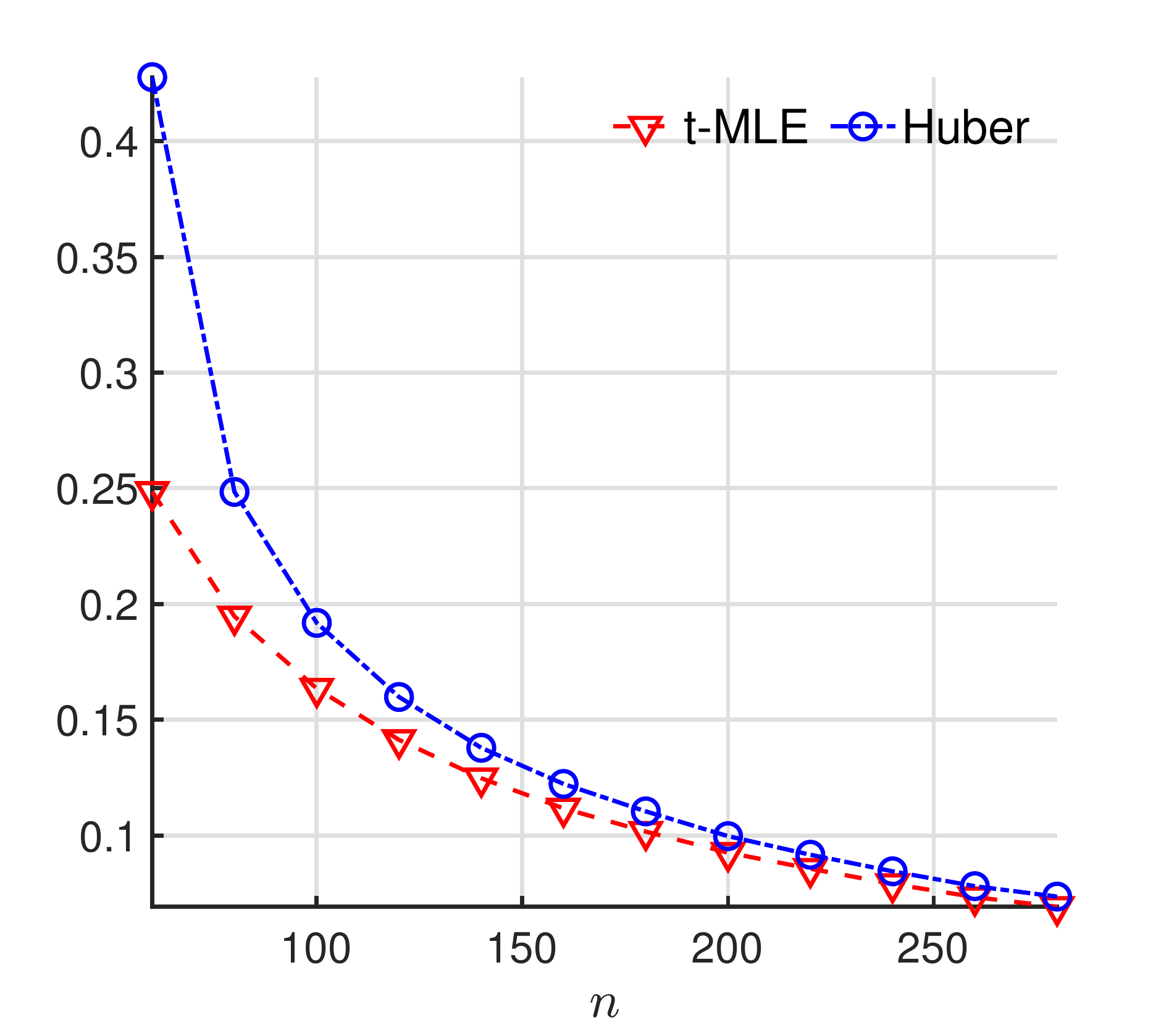}} 
\vspace{-0.4cm}
\caption{ NMSE of the estimators as a  function of $n$  when samples are draw from a $p$-variate $t_5$ (left panel) and $t_3$ (right panel) distribution with an AR(1) covariance structure; $\varrho=0.6$ and $p=40$.} \label{fig:AR1_t}
\centerline{\includegraphics[width=0.37\textwidth]{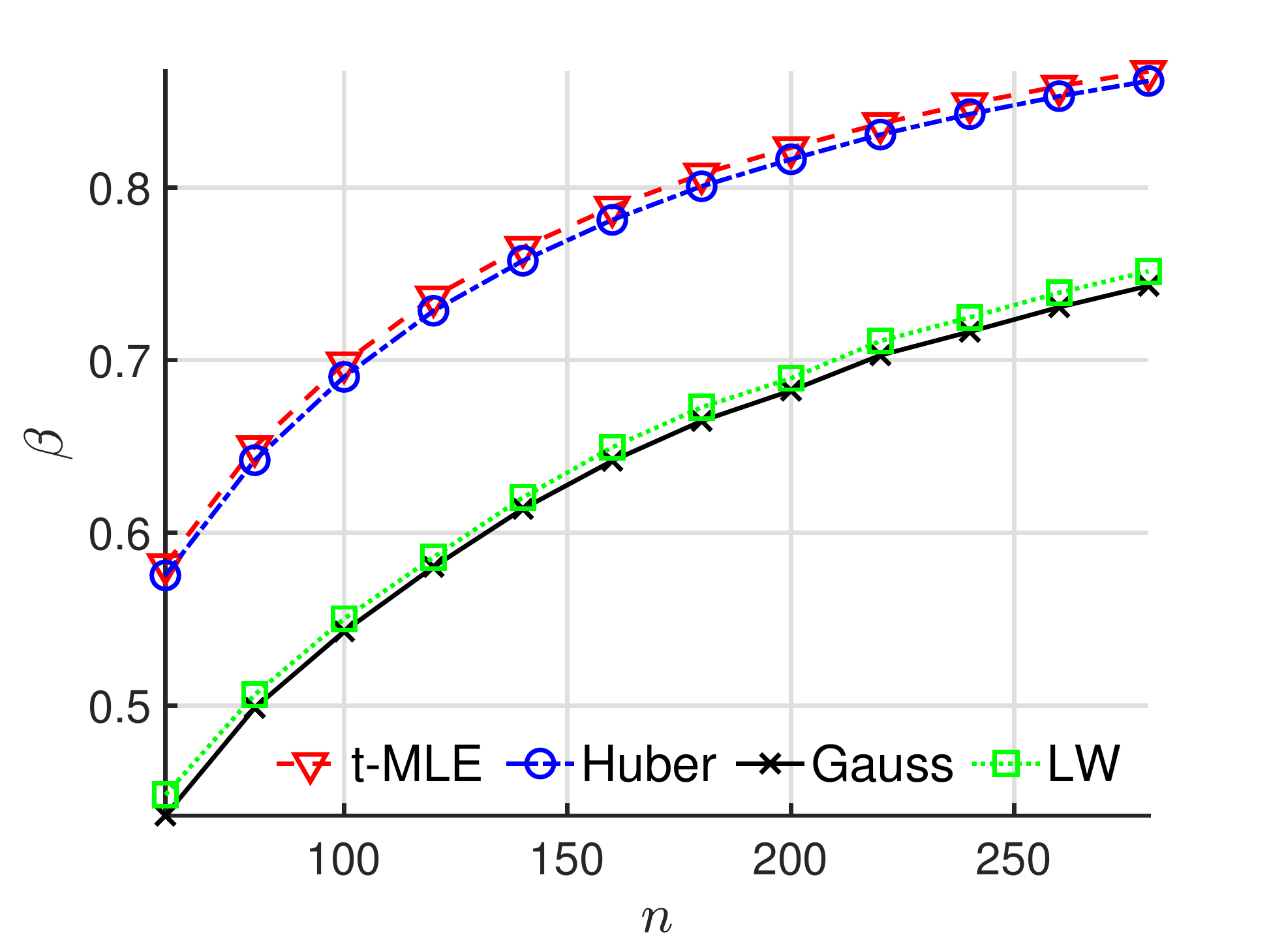}}
\vspace{-.4cm}
\caption{ Shrinkage parameter $\beta$  as a function of $n$  when samples are drawn from a  $p$-variate $t_5$-distribution with an  AR(1) covariance structure; $\varrho=0.6$ and $p=40$.} \label{fig:AR1_beta}\vspace{-.2cm}
\end{figure}

 We generated the data from an elliptical distribution $\mathcal E_p(\bo 0, \M,g)$, where 
the   scatter matrix  $\M$ has an AR(1) structure, $(\M)_{ij} =  \eta \varrho^{|i-j|}$, where $\varrho \in (0,1)$ and 
 scale  parameter  $\eta = \tr(\M)/\pdim =10$.  When $\varrho
\downarrow 0$, then $\M$ is close to an identity matrix scaled by $\eta$, and when $\varrho
\uparrow 1$, $\M$ tends to a singular matrix of rank 1.  Parameter $\varrho$ is set to $\varrho=0.6$. 
The dimension is $\pdim = 40$ and $\ndim$ varies from 60 to 280. 

In our first study, samples are drawn from a MVN distribution and  the normalized MSE $\| \hat \M_\be - \Mn \|_{\Fr}^2/\| \Mn \|_{\Fr}^2$  as a function of sample length $n$ is depicted in Figure~\ref{fig:AR1}.  Results are averages over 2000 Monte-Carlo trials. All estimators are performing well;  Gauss and t-MLE are performing slightly better than LW or Huber but differences are marginal. 

Figure~\ref{fig:AR1_t}  shows the NMSE figures in the case that samples are from $t_5$- and $t_3$-distribution, respectively. In the latter case, the non-robust Gauss and LW estimator provided large NMSE and are therefore not shown in the plot. This was expected as $t_3$-distribution is heavy-tailed with non-finite kurtosis. As can be seen, the robust Huber and t-MLE shrinkage estimators  provide significantly improved performance when the data is sampled from a heavy-tailed $t_5$ or $t_3$-distribution. We can also notice that t-MLE estimator that adaptively estimates the d.o.f. $ \nu$ from the data is able to outperform the Huber's M-estimator due to the data adaptivity. 

Figure~\ref{fig:AR1_beta} depicts the (average) shrinkage parameter $\beta$ as a function of $n$ in the case that samples are from a $p$-variate $t_5$ distribution.  As can be seen the robust shrinkage estimators (Huber and t-MLE) use larger shrinkage parameter value $\beta$ than Gauss and LW.  

\section{Conclusions and perspectives}

This work proposed an original and fully automatic approach to compute an optimal shrinkage parameter in the context of heavy-tailed distributions and/or in presence of outliers. It has been shown that the performance of the method is similar to the optimal one when the data is Gaussian while it outperforms shrinkage Gaussian-based methods when the data distribution turns out to be non-Gaussian. This paper opens several ways, notably considering the case when $p>n$.

\end{document}